\documentclass[submission,copyright,creativecommons]{eptcs}
\providecommand{\event}{PC 2016} 
\usepackage{breakurl}             

\usepackage{amsmath,amsthm,amsfonts,amssymb}
\usepackage{fontenc}
\usepackage[all]{xy}
\usepackage{graphicx,subfigure,multirow}
\usepackage[figurename=Fig.]{caption}
\usepackage{tikz}
\usetikzlibrary{shapes, arrows, shadows}
\usepackage[latin1]{inputenc}
\usepackage[scaled]{helvet}
\usepackage[toc, page]{appendix}
\usepackage[sort]{cite}

\newtheorem{theorem}{Theorem}[subsection]
\newtheorem{definition}[theorem]{Definition}

\numberwithin{equation}{subsection} 

\usepackage{hyperref}

\title{The Information Content of Systems in General Physical Theories}
\author{Ciar\'{a}n M. Lee \qquad\qquad Matty J. Hoban
\institute{University of Oxford\\Department of Computer Science\\Wolfson Building\\Parks Road\\Oxford OX1 3QD, UK}
\email{\{ciaran.lee,matthew.hoban\}@cs.ox.ac.uk}
}
\def\titlerunning{The Information Content of Systems in General Physical Theories}
\def\authorrunning{C.M. Lee, \& M.J. Hoban}
\begin{document}
\maketitle

\begin{abstract}
What kind of object is a quantum state? Is it an object that encodes an exponentially growing amount of information (in the size of the system) or more akin to a probability distribution? It turns out that these questions are sensitive to what we do with the information. For example, Holevo's bound tells us that $n$ qubits only encode $n$ bits of classical information but for certain communication complexity tasks there is an exponential separation between quantum and classical resources. Instead of just contrasting quantum and classical physics, we can place both within a broad landscape of physical theories and ask how non-quantum (and non-classical) theories are different from, or more powerful than quantum theory. For example, in communication complexity, certain (non-quantum) theories can trivialise all communication complexity tasks. In recent work [C. M. Lee and M. J. Hoban, Proc. Royal Soc. A 472 (2190), 2016], we showed that the immense power of the information content of states in general (non-quantum) physical theories is not limited to communication complexity. We showed that, in general physical theories, states can be taken as ``advice" for computers in these theories and this advice allows the computers to easily solve any decision problem. Aaronson has highlighted the close connection between quantum communication complexity and quantum computations that take quantum advice, and our work gives further indications that this is a very general connection. In this work, we review the results in our previous work and discuss the intricate relationship between communication complexity and computers taking advice for general theories. 
\end{abstract}

\section{Introduction}    

Quantum theory holds the promise of more powerful algorithms and securer communication \cite{shor}. In turn, these possibilities have affected the kinds of questions we ask about quantum theory. In particular, if quantum theory was replaced with another theory, what would the information processing consequences be \cite{Barrett-2007,Popescu-review,PR-van-Dam}? By asking these sorts of questions we can understand quantum theory better through its limitations as well as its strengths, and this understanding will allow us to maximise its potential. 

One oft-asked question in the foundations of quantum theory is what kind of object is the quantum state \cite{Aaronson-advice}? Is it like a classical probability distribution or an exponentially long vector \cite{PBR}? For $n$ qubits, $2^n$ coefficients are required, in general, to describe the state of the system yet Holevo's theorem tells us that only $n$ classical bits can be reliably encoded into the system \cite{holevo}. Clearly, the answer to the question is sensitive to the context in which it is asked. 

One concrete context in which we can ask about the information content of quantum states is in the study of communication complexity \cite{Buhrman-review}. There are many varieties of communication complexity and depending on the variety, there is no separation between classical and quantum resources \cite{Nielsen-comm} or there is an \textit{exponential} separation between randomised classical two-way communication and one-way quantum communication \cite{Klartag}. That is, for certain tasks, to classically simulate the sending of a quantum state from Alice to Bob requires an exponential amount of (even two-way) randomised classical communication. In this way, the quantum state seems like something very different from a classical probability distribution. 

So if a quantum state is not a probability distribution, then what is it? An approach to answering this question is to devise a general framework of theories that both includes classical and quantum theory as examples and also makes good operational sense. Luckily, such frameworks have been proposed \cite{Barrett-2007,Pavia1,Pavia2,Hardy-2011} and have laid the path for an impressive range of results. Within this framework, then, we can compare the information content of quantum states with the information content of states within general theories. We can ask how the information content of a state depends on the underlying physical features of the theory and viewing quantum theory in this more general context can yield insight into the nature of the quantum state. Returning to the theme of communication complexity, for a particular task there is a vast difference between an arbitrary theory and quantum theory. For a theory colloquially known as ``Boxworld" \cite{Barrett-2007}, communication complexity tasks can be rendered completely trivial \cite{PR-van-Dam}. Given this perspective, states in this theory are vastly more powerful than in quantum theory.

The result of trivial communication complexity for Boxworld has motivated the non-triviality of communication complexity as an information theoretic principle that could pick out quantum theory, or at least some subset of all theories \cite{Brassard}. In the restricted, but very related setting of studying non-locality, it has been shown that there exists a consistent set of non-quantum correlations (called the ``almost quantum correlations") that does not lead to trivial communication complexity so this principle cannot single out quantum theory \cite{almostq}. It is then natural to ask what are the theories that look like quantum theory from the perspective of communication complexity and do they share very common structure?

One difficulty with studying communication complexity is in the variety of different scenarios and resources that can be studied. As highlighted above, for one task, one can have an exponential separation between quantum and classical, and no separation at all for another. There are also complications in translating between scenarios for different theories. For example, in quantum communication complexity, due to teleportation we can translate between the setting of having only communication of qubits to the setting of having pre-shared entanglement and only classical communication \cite{Buhrman-review}. Theories more general than quantum theory may not permit teleportation so certain comparisons can seem unfair. Boxworld with reversible dynamics does not permit teleportation \cite{Rev}, and so even though communication complexity is rendered trivial in the case with pre-shared Boxworld correlations, it's not clear if every protocol of this form can be simulated by communicating only a constant number of systems in Boxworld without pre-shared correlations. We need a clear framework in which we can ask general questions about a theory that does not make too many assumptions about resource interconversion within a theory. 

In this direction, we look at the computational complexity of circuits that take advice. This gives a general framework that can address the question of how much information can be encoded in a state within a general theory. As Aaronson pointed out \cite{Aaronson-advice}, this framework is closely related to the setting of one-way communication complexity so we can gain insight into the latter by studying the former. We will further elaborate on the connections between the two. In particular, we show that an argument demonstrating that communication complexity is trivial in Boxworld can also be used to demonstrate the computational complexity of Boxworld circuits that take advice. Going further, we non-trivially bound the computational complexity of circuits that take advice for a general class of theories satisfying natural assumptions. We then comment on how this result might be used to classify theories with non-trivial communication complexity. The work presented here is based on a general discussion and technical results in \cite{LeeHoban} but now with expanded discussion from the perspective of communication complexity.

\section{Circuits with advice in general physical theories}

\subsection{Operational theories} 

We work in the circuit framework for generalised probabilistic theories developed by Hardy in \cite{Hardy-2011} and Chiribella, D'Ariano and Perinotti in \cite{Pavia1,Pavia2}. The presentation here is most similar to that of Chiribella \textit{et al}. We now provide a brief review of this framework, see \cite{LB-2014, LeeHoban} for more in-depth reviews and an extended discussion of computation in general theories.  

A theory within this framework specifies a set of laboratory devices that can be connected together in different ways to form experiments and assigns probabilities to different experimental outcomes. Each device has a classical pointer indicating an event that has occurred. In a general theory, one can depict the connections of devices in some experimental set-up by closed circuits. A requirement on any theory is that it should give probabilistic predictions about the occurrence of possible outcomes (i.e. the value of the classical pointer). It is thus demanded that, in this framework, closed circuits define probability distributions. Given this structure, one then says that two physical devices are equivalent (from the point of view of the theory) if replacing one by the other in any closed circuit does not change the probabilities. The set of equivalence classes of devices with no input ports are referred to as \emph{states}, devices with no output ports as \emph{effects} and devices with both input and output ports as \emph{transformations}. 

The notation $|s_{r})_A$ is used to represent a state of system type $A$, where $r$ is the outcome of the classical pointer, and $_A(e_{r}|$ to represent an effect on system type $A$, so that if the effect $_A(e_{r_2}|$ is applied to the state $|s_{r_1})_A$, the probability to obtain outcome $r_1$ on the physical device representing the state and outcome $r_2$ on the physical device representing the effect is $_A(e_{r_2}|s_{r_1})_A := P(r_1,r_2)$. The fact that closed circuits correspond to probabilities can be leveraged to show that the set of states, effects and transformations each give rise to a vector space and that the transformations and effects act linearly on the vector space of states. We assume in this work that all vector spaces are finite dimensional.

We can now formally define some examples of physical principles. 
\begin{definition}[Causality \cite{Pavia1}]
A theory is said to be \emph{causal} if the marginal probability of a preparing a state is independent of the choice of which measurement follows the preparation. 
\end{definition}

\begin{definition}[Tomographic locality \cite{Pavia1, Hardy-2011, Barrett-2007}]
A theory satisfies tomographic locality if every transformation can be uniquely characterised by local process tomography. Local process tomography is the act of collecting statistics from only inputting local, product states into a process and only making local measurements. 
\end{definition}

We will now define the principle of bit-symmetry. Before we define this principle, the following concepts must be introduced. We say the laboratory device $\{\mathcal{U}_j\}_{j\in{Y}}$, where $j$ indexes the positions of the classical pointer, is a \emph{coarse-graining} of the device $\{\mathcal{E}_i\}_{i\in{X}}$ if there is a disjoint partition $\{X_j\}_{j\in{Y}}$ of $X$ such that $\mathcal{U}_j=\sum_{i\in{X_j}}\mathcal{E}_i$. That is, coarse-graining arises when some outcomes of a laboratory device are joined together. The device $\{\mathcal{E}_i\}_{i\in{X}}$ is said to \emph{refine} the device $\{\mathcal{U}_j\}_{j\in{Y}}$. A state is \emph{pure} if it does not arise as a \emph{coarse-graining} of other states; a pure state is one for which we have maximal information. A state is \emph{mixed} if it is not pure and it is \emph{completely mixed} if any other state refines it. That is, $|c)$ is completely mixed if for any other state $|\rho)$, there exists a non-zero probability $p$ such that $p|\rho)$ refines $|c)$. States $\{|\sigma_i)\}_{i=1}^N$ are \emph{perfectly distinguishable} if there exists a measurement, corresponding to effects $\{(e_i|\}_{i=1}^N$, such that $(e_i|\sigma_j)=\delta_{ij}$ for all $i,j$. 

\begin{definition}[Bit-symmetry \cite{MU}] \label{symm}
A theory satisfies \emph{bit-symmetry} if for any two $2$-tuples of pure and perfectly distinguishable states $\{|\rho_1),|\rho_2)\},\{|\sigma_1), |\sigma_2)\},$ there exists a reversible transformation $T$ such that $T|\rho_i)=|\sigma_i)$ for $i=1,2$.
\end{definition}

Note that causality, tomographic locality and bit-symmetry are all logically independent: generalised probabilistic theories satisfying any subset (including the empty subset) can be defined. For example, standard quantum theory satisfies all three, quantum theory with real amplitudes satisfies causality and bit-symmetry but not tomographic locality, Boxworld satisfies causality and tomographic locality but not bit-symmetry \cite{MU} and the theory constructed in \cite{cause} does not satisfy causality. 

\subsection{Efficient circuits that take advice} \label{general-def} 

To define the class of efficient computation in a general theory, we must first define the notions of a uniform circuit family and an acceptance condition for an arbitrary theory. The notion of a poly-size uniform circuit family $\{C_x\}$, which is indexed by some bit string $x$ is defined in \cite{LB-2014}. In this definition, a classical Turing machine gives an efficient description of a circuit, and the classical outcomes associated with the pointers on the devices are efficiently processed by this classical Turing machine to give a classical output (acceptance or rejection). 

In the paradigm of uniform circuits that take advice, one is given both the problem instance $x$ and an advice state, so the constructed circuit $C_x$ must have open system ports into which this state can be plugged. Henceforth we will assume that uniform circuit families consist of collections of circuits with a number of open input ports, which can grow as a polynomial in $|x|$, which we call the \emph{auxiliary register}. Note that the choice of finite gate set determines the possible system types of the auxiliary register. Given this convention, we can define efficient computation with trusted advice in a specific general theory.


\begin{definition}
For a general theory $\bold{G}$, a language $\mathcal{L}\subseteq\{0,1\}^{n}$ is in the class $\bold{BGP/gpoly}$ if there exists a poly-sized uniform family of circuits $\{C_x\}$ in $\bold{G}$, a set of (possibly non-uniform) states $\{\sigma_{n}\}_{n\geq 1}$ on a composite system of size $d(n)$ for some polynomial $d:\mathbb{N}\rightarrow\mathbb{N}$, and an efficient acceptance criterion, such that for all strings $x\in\{0,1\}^n$:
\begin{enumerate}
\item If $x\in\mathcal{L}$ then $C_x$ accepts with probability at least $2/3$ given $\sigma_{n}$ as input to the auxiliary register.

\item If $x\notin\mathcal{L}$ then $C_x$ accepts with probability at most $1/3$ given $\sigma_{n}$ as input to the auxiliary register.
\end{enumerate}
\end{definition}

Here by ``composite system of size $d(n)$'', we mean that the number of systems, or open ports, of the auxiliary register -- into which the advice state is input -- increases as $d(n)$, for $d$ a polynomial in the input size. The constants $(\frac{2}{3},\frac{1}{3})$ can be chosen arbitrarily as long as they are bounded away from $\frac{1}{2}$ by some constant. The example of $\bold{BGP/gpoly}$ for quantum theory, called $\bold{BQP/qpoly}$ was introduced by Nishimura and Yamakami \cite{NY03}. The classical version of this class is known to be equal to $\bold{P/poly}$, the class of deterministic, classical Turing machines that take advice.

We now look at Boxworld with respect to our definitions advice in general physical theories. Towards this end we provide a brief definition of Boxworld, see e.g. \cite{PR-trade-off} for a more in-depth discussion. For a given single system $A$ in Boxworld, there are two choices of binary-outcome measurements, $\{_A(x_a|\}$ for $x,a\in\{0,1\}$. Here $x$ is the bit denoting the two possible choices of measurement and $a$ is the bit denoting the two possible outcomes of the chosen measurement, i.e the two measurements on system $A$ are $\{_A(0_0|, _A(0_1|\}$ and $\{_A(1_0|, _A(1_1|\}$. States and measurements in this theory can produce correlations associated with the so-called Popescu-Rohrlich non-local box \cite{PRbox}. These bipartite correlations can be extended to an $n$-partite system where now for the $j$th party, $x_{j}\in\{0,1\}$ and $a_{j}\in\{0,1\}$ are the choice of measurement and its outcome respectively. There exists a state $|\rho_{f})$ and effects $\{_j(x_{j},a_{j}|\}$ for all $j$ parties that produce the probabilities
$$(x_{1},a_{1}|(x_{2},a_{2}|...(x_{n},a_{n}|\rho_f)= 
\begin{cases}
	\frac{1}{2^{n-1}} & \text{if $\bigoplus_{j=1}^{n}a_{j}= f(x)$,}\\
	0 & \text{otherwise,}
\end{cases}
$$
where $\bigoplus$ represents summation modulo $2$ and $f:\{0,1\}^{n}\rightarrow\{0,1\}$ is any Boolean function from the bit-string $x$ with elements $x_{j}$. Therefore, if the state $|\rho_{f})$ is prepared and local measurements described by effects $(x_{j},a_{j}|$ made, a classical computer can compute the parity of all outcomes $a_{j}$ and so we deterministically obtain the evaluation of a Boolean function $f(x)$. This relatively straightforward observation gives us the following result.

\begin{theorem}\label{thmboxworld}
\cite{LeeHoban} There exist generalised probabilistic theories  $\bold{G}$ satisfying causality and tomographic locality, which satisfy $\bold{BGP/gpoly}=\bold{ALL}$ where $\bold{ALL}$ is the class of all decision problems.
\end{theorem}

\begin{proof}
Clearly $\bold{BGP/gpoly}\subseteq\bold{ALL}$ is trivially true for Boxworld. The states $|\rho_{f})$ can be used as advice states and, as all decision problems can be represented by Boolean functions, it follows that $\bold{ALL}\subseteq\bold{BGP/gpoly}$.
\end{proof}

It was first established by Aaronson that $\bold{BQP/qpoly}\subseteq\bold{PP/poly}\subsetneq \bold{ALL}$ thus quantum mechanical states cannot encode the answers to all problems, unlike the case for Boxworld. So, clearly, we need more principles than causality and tomographic locality to give theories that have non-trivial upper bounds on the computational power of advice. In the following result, we show that the principle of bit-symmetry is such a principle and a theory satisfying it (recall that Boxworld does not) cannot use advice to solve all decision problems.

\begin{theorem}\label{advicetheorem}
\cite{LeeHoban} Any causal, bit-symmetric, tomographically local theory $\bold{G}$ with at least two pure and distinguishable states satisfies
$$\bold{BGP/gpoly} \subseteq \bold{PP/poly} \subsetneq \bold{ALL}.$$
\end{theorem}

\section{Connections between advice and communication complexity}

Earlier, we discussed the information content of states within general theories from the point-of-view of communication complexity. However, the framework and results in this work were phrased in terms of computations that take advice. This framework allows us to concretely ask how much information content there is in a state and we showed that informational principles can limit the information content of states. We now end with some comments on the connection between advice and communication complexity. A pertinent question is whether our results can say something about communication complexity. For example, if $\bold{BGP/gpoly}=\bold{ALL}$ for some theory $\bold{G}$, does this mean all communication complexity tasks are trivial in this theory? 

In the case of Boxworld, we can adapt the proof that $\bold{BGP/gpoly}=\bold{ALL}$ to the communication complexity scenario. In such a scenario, Alice has an input bit-string $x\in\{0,1\}^n$ and Bob has $y\in\{0,1\}^n$ and they wish to perform some function $f(x,y)$. They are allowed to share arbitrary states and correlations in a theory prior to receiving the inputs but after receiving the inputs they can only classically communicate. In the case of Boxworld, for a particular function $f(x,y)$, they prepare the state $\vert\rho_{f(x,y)})$ described above and the first $n$ systems are held by Alice whereas the second $n$ systems are held by Bob. Upon receiving the inputs $x$ and $y$ respectively, they make measurements corresponding to these input choices, Alice then takes the parity of her outputs and sends this bit value to Bob. Bob takes the parity of his outputs with the bit that Alice sends and gets $f(x,y)$ with certainty. Any such task can be achieved through only communicating one classical bit. 

Can such a mapping be made more general? Is communication complexity non-trivial in theories where $\bold{BGP/gpoly} \subseteq \bold{PP/poly}$? We conjecture that this second question has a positive answer. One can give a bound on one-way communication complexity (without pre-shared ``correlated" systems) in general theories that is similar to one that Aaronson proves in \cite{Aaronson-advice}, but for theories satisfying certain non-trivial properties such as purification \cite{Pavia1,Pavia2} which are not required in Theorem \ref{advicetheorem}. In general, the connection between communication complexity and computations that take advice is not as straightforward as in the example of Boxworld. For example, it may be the case that a communication complexity task could be rendered trivial but only when two parties share an exponential amount of resources in the size of the classical input. Van Dam's original protocol was of this form \cite{PR-van-Dam}, and even though we can improve its ``efficiency", this may not be possible in general. If the state that Alice and Bob have to share is exponentially large (the number of sub-systems the parties have is exponential in the size of the input) then this is not a viable advice state according to our definitions. What we can say is, if there exists an efficient protocol (the states used and runtime are polynomial in the size of the input) within a theory that trivialises communication complexity, then $\bold{BGP/gpoly}=\bold{ALL}$. Another possible indication of the connection between the two might be that the proof of Theorem \ref{advicetheorem} can be modified to derive a bound on one-way communication complexity (without prior correlations) in a theory satisfying the same principles in the statement of the theorem. 

In this work, we have related our work in \cite{LeeHoban} to the study of communication complexity in general physical theories. There has been some prior work in this direction studying the one-way communication complexity of general theories with some initially intriguing and nice results \cite{massar1,massar2}. In future work we would like to connect this study to the study of computations that take advice. In drawing these threads together we may understand why Nature chose quantum theory and not some other possibility, and in doing so, we might understand what kind of object is the quantum state.

\end{document}